\documentclass[10pt,conference,comsocconf]{IEEEtran}

\usepackage[acronym]{glossaries}
\usepackage[T1]{fontenc}

\usepackage{blindtext}
\usepackage{graphicx}
\usepackage{subcaption}
\usepackage{placeins} 
\usepackage{amsmath,amssymb}
\usepackage{braket}

\usepackage{amsthm}
\newtheorem{lemma}{Lemma}

\usepackage{multirow} 
\usepackage{tikz} 
\usetikzlibrary{calc,fit,shapes}
\usetikzlibrary{hobby,backgrounds,trees}
\usepgflibrary{arrows}
\usepackage{varwidth}
\usetikzlibrary{positioning,arrows}

\usepackage{url}     
\usepackage[pdftex,
    hidelinks,
    bookmarks=false,
    pdfpagemode=UseNone
]{hyperref}

\makeatletter
\AtBeginDocument{
  \hypersetup{
    draft=true
  }
}
\makeatother

\usepackage{fancyhdr}

\usepackage{multirow}

\usepackage{tcolorbox}

\usepackage{tabularx}

\begin{document}


\thispagestyle{empty} 

\vfill 

\begin{center}

\begin{tcolorbox}[
    colback=gray!10, colframe=black, fonttitle=\bfseries,
    width=0.9\textwidth, boxrule=1pt, arc=5pt, outer arc=5pt,
    boxsep=10pt, left=10pt, right=10pt, top=10pt, bottom=10pt
]
\textbf{THIS IS AN AUTHOR-CREATED POSTPRINT VERSION.}

\vspace{0.3cm}

\textbf{Disclaimer:}  
This work has been accepted for publication in the \textit{Joint European Conference on Networks and Communications \& 6G Summit (EuCNC/6G Summit)}, 2026.  

\vspace{0.3cm}

\textbf{Copyright:}  
© 2026 IEEE. Personal use of this material is permitted. Permission from IEEE must be obtained for all other uses, in any current or future media, including reprinting/republishing this material for advertising or promotional purposes, creating new collective works, for resale or redistribution to servers or lists, or reuse of any copyrighted component of this work in other works.

\vspace{0.3cm}


\end{tcolorbox}
\end{center}

\vfill 

\clearpage
\setcounter{page}{1}

\newacronym{3GPP}{3GPP}{3rd Generation Partnership Project}

\newacronym{5G}{5G}{5th Generation}
\newacronym{5G-ACIA}{5G-ACIA}{5G Alliance for Connected Industries and Automation}
\newacronym{5G-NR}{5G-NR}{5G New Radio}
\newacronym{6G}{6G}{Sixth Generation}

\newacronym{AMC}{AMC}{Adaptive Modulation and Coding}
\newacronym{AC}{AC}{Admission Control}
\newacronym{AGV}{AGV}{Automated Guided Vehicle}
\newacronym{AR}{AR}{Augmented Reality}

\newacronym{BLER}{BLER}{Block Error Rate}
\newacronym{BWP}{BWP}{Bandwidth Part}
\newacronym{BS}{BS}{Base Station}
\newacronym{BSS}{BSS}{Business Support System}
\newacronym{BSM}{BSM}{Bell State Measurement}

\newacronym{CDF}{CDF}{Cumulative Distribution Function}
\newacronym{CCDF}{CCDF}{Complementary Cumulative Distribution Function}
\newacronym{CDMA}{CDMA}{Code Division Multiple Access}
\newacronym{CTMC}{CTMC}{Continuos-Time Markov Chain}
\newacronym{CSI}{CSI}{Channel State Information}
\newacronym{CP}{CP}{Control Plane}
\newacronym{CQI}{CQI}{Channel Quality Indicator}
\newacronym{CU}{CU}{Centralized Unit}

\newacronym{DL}{DL}{downlink}
\newacronym{DNC}{DNC}{Deterministic Network Calculus}
\newacronym{DRP}{DRP}{Dynamic Resource Provisioning}
\newacronym{DRL}{DRL}{Deep Reinforcement Learnning}
\newacronym{DU}{DU}{Distributed Unit}

\newacronym{eMBB}{eMBB}{enhanced Mobile Broadband}
\newacronym{ETSI}{ETSI}{European Telecommunication Standards Institute}
\newacronym{EBB}{EBB}{Exponentially Bounded Burstiness}
\newacronym{EBF}{EBF}{Exponentially Bounded Fluctuation}
\newacronym{E2E}{E2E}{End-to-End}
\newacronym{EDF}{EDF}{Earliest Deadline First}
\newacronym{EM}{EM}{Expectation-Maximization}

\newacronym{FCFS}{FCFS}{First-come First-served}
\newacronym{FIFO}{FIFO}{First In First Out}
\newacronym{FSO}{FSO}{Free-Space Optical}

\newacronym{GBR}{GBR}{Guaranteed Bit Rate}
\newacronym{GMM}{GMM}{Gaussian Mixture Model}
\newacronym{GSMA}{GSMA}{Global System for Mobile Communications Association}
\newacronym{GST}{GST}{Generic Network Slice Template}
\newacronym{gNB}{gNB}{Next generation NodeB}

\newacronym{HDR}{HDR}{High Data Rate}

\newacronym{ITU}{ITU}{International Telecommunication Union}
\newacronym{IoT}{IoT}{Internet of Things}
\newacronym{ILP}{ILP}{Integer Linear Programming}
\newacronym{ICIC}{ICIC}{Inter-Cell Interference Cancellation}

\newacronym{LA}{LA}{Link Adaptation}
\newacronym{LOS}{LoS}{Line-of-Sight}
\newacronym{LSTM}{LSTM}{Long Short-Term Memory}
\newacronym{LTE}{LTE}{Long Term Evolution}

\newacronym{MAC}{MAC}{Medium Access Control}
\newacronym{MEC}{MEC}{Multi-access Edge Computing}
\newacronym{MCS}{MCS}{Modulation and Coding Scheme}
\newacronym{MDN}{MDN}{Mixture Density Network}
\newacronym{MGF}{MGF}{Moment Generating Function}
\newacronym{MIMO}{MIMO}{Multiple Input Multiple Output}
\newacronym{MISO}{MISO}{Multiple Input Single Output}
\newacronym{MRC}{MRC}{Maximal Ratio Combining}
\newacronym{ML}{ML}{Machine Learning}
\newacronym{MNO}{MNO}{Mobile Network Operator}
\newacronym{mMTC}{mMTC}{Machine Type Communication}
\newacronym{MSE}{MSE}{Mean Squared Error}
\newacronym{mURLLC}{mURLLC}{massive ultra-Reliable Low Latency Communication}

\newacronym{NE}{NE}{Nash Equilibrium}
\newacronym{NEST}{NEST}{Network Slice Type}
\newacronym{NIP}{NIP}{Non-linear Integer Programming}
\newacronym{NFMF}{NFMF}{Network Function Management Function}
\newacronym{NFV}{NFV}{Network Function Virtualization}
\newacronym{NG-RAN}{NG-RAN}{Next Generation - RAN}
\newacronym{NLOS}{NLoS}{Non-Line-of-Sight}
\newacronym{NN}{NN}{Neural Network}
\newacronym{NSO}{NSO}{Network Slice Orchestrator}
\newacronym{NSMF}{NSMF}{Network Slice Management Function}
\newacronym{NSSMF}{NSSMF}{Network Slice Subnet Management Function}
\newacronym{NR}{NR}{New Radio}
\newacronym{NV}{NV}{Nitrogen-Vacancy}

\newacronym{OFDMA}{OFDMA}{Orthogonal Frequency-Division Multiple Access}
\newacronym{O-RAN}{O-RAN}{Open Radio Access Network}
\newacronym{PDF}{PDF}{Probability Density Function}
\newacronym{PMF}{PMF}{Probability Mass Function}
\newacronym{PRB}{PRB}{Physical Resource Block}
\newacronym{P-NEST}{P-NEST}{private NEST}

\newacronym{QApp}{QApp}{Quantum Application}
\newacronym{QoS}{QoS}{Quality of Service}
\newacronym{QN}{QN}{Quantum Network}
\newacronym{QR}{QR}{Quantum Repeater}

\newacronym{RAN}{RAN}{Radio Access Network}
\newacronym{RB}{RB}{Resource Block}
\newacronym{RBG}{RBG}{Resource Block Group}
\newacronym{RIC}{RIC}{RAN Intelligent Controller}
\newacronym{RIS}{RIS}{Reconfigurable Intelligent Surfaces}
\newacronym{RRM}{RRM}{Radio Resource Management}
\newacronym{RSMA}{RSMA}{Rate Splitting Multiple Access}
\newacronym{RSRP}{RSRP}{Received Signal Received Power}
\newacronym{RSRQ}{RSRQ}{Received Signal Received Quality}
\newacronym{RSSI}{RSSI}{Received Signal Strength Indication}
\newacronym{RT}{RT}{Real Time}
\newacronym{RU}{RU}{Radio Unit}

\newacronym{SDO}{SDO}{Standards Developing Organization}
\newacronym{SINR}{SINR}{Signal-to-Interference-plus-Noise Ratio}
\newacronym{SLA}{SLA}{Service Level Agreement}
\newacronym{SNC}{SNC}{Stochastic Network Calculus}
\newacronym{SNR}{SNR}{Signal-to-Noise Ratio}
\newacronym{S-NEST}{S-NEST}{standardized NEST}

\newacronym{TTI}{TTI}{Transmission Time Interval}

\newacronym{UE}{UE}{User Equipment}
\newacronym{UL}{UL}{Uplink}
\newacronym{UP}{UP}{User Plane}
\newacronym{uRLLC}{uRLLC}{ultra-Reliable Low Latency Communication}

\newacronym{V2X}{V2X}{Vehicle-to-Everything}
\newacronym{VBQC}{VBQC}{Verifiable Blind Quantum Computing}
\newacronym{VBR}{VBR}{Variable Bit Rate}
\newacronym{VR}{VR}{Virtual Reality}
\newacronym{vRAN}{vRAN}{virtualized RAN}
\newacronym{vBS}{vBS}{virtualized Base Station}
\newacronym{VS}{VS}{Validation Scenario}

\newacronym{WiMAX}{WiMAX}{Worldwide Interoperability for Microwave Access}
\newacronym{WCDMA}{WCDMA}{Wideband \gls{CDMA}}

%
\title{How Many Qubits Can Be Teleported? Scalability
of Fidelity-Constrained Quantum Applications}

\author{\IEEEauthorblockN{Oscar Adamuz-Hinojosa, Jonathan Prados-Garzon, Sara Vaquero-Gil, Juan M. Lopez-Soler}

\IEEEauthorblockA{Department of Signal Theory, Telematics and Communications, University of Granada.}
\IEEEauthorblockA{Email: \{oadamuz,jpg,juanma\}@ugr.es; saravaquerogil@correo.ugr.es}
}

\markboth{Journal of \LaTeX\ Class Files,~Vol.~6, No.~1, January~2007}%
{Shell \MakeLowercase{\textit{et al.}}: Bare Demo of IEEEtran.cls for Journals}

\maketitle

\begin{abstract}
Quantum networks (QNs) enable qubit transfer between distant nodes through quantum teleportation, which reconstructs a quantum state at a remote node by consuming a shared Bell pair. In multi-qubit quantum applications (QApps), the teleported qubits may need to remain stored in quantum memories until execution can start, while decoherence progressively reduces their fidelity with respect to the ideal target state. Such QApps can operate only if all teleported qubits simultaneously satisfy a minimum fidelity threshold. In this paper, we study how many qubits can be teleported under this fidelity-constrained operation in a two-node QN. To this end, we define a QApp-level reliability metric as the probability that all end-to-end Bell pairs satisfy the target fidelity when the multi-qubit teleportation stage is completed. We then develop a Monte Carlo simulator that captures stochastic Bell-pair generation, Quantum Repeater (QR)-assisted entanglement distribution, and fidelity degradation. The analysis considers fiber-based and terrestrial free-space optical (FSO) quantum links, as well as representative NV-center- and trapped-ion-based quantum memories. Results show that memory coherence is the main scalability bottleneck under stringent fidelity targets, while parallel entanglement generation is essential for multi-qubit teleportation.
\end{abstract}


\section{Introduction}
Quantum information is encoded in fragile quantum states that cannot be copied or regenerated without disturbance. As a result, long-distance qubit transmission is strongly limited by loss and noise, and classical amplification cannot be used to preserve quantum information~\cite{Wehner_QuantumInternet18,pirandola2017fundamental}. To overcome this limitation, \glspl{QN} rely on quantum teleportation to transfer a quantum state from a source node (Alice) to a destination node (Bob) without physically sending the qubit itself~\cite{Bouwmeester1997}. Teleportation consumes pre-shared entanglement, typically in the form of Bell pairs, together with classical communication to reconstruct the state at Bob.

Many \glspl{QApp} include an execution-gated phase in which multiple qubits must be teleported from Alice to Bob before quantum processing can begin. This requires the prior establishment of multiple end-to-end Bell pairs. Since Bell-pair generation is probabilistic, these pairs become available at different times, so earlier ones must remain stored in quantum memories while later ones are established. During this storage time, decoherence degrades the corresponding quantum states. Execution is therefore feasible only if, when the teleportation phase is completed, all required qubits are simultaneously available at Bob with fidelity above a minimum threshold, where fidelity measures the overlap between the stored state and the ideal target state~\cite{davies24tools}. Representative examples of such execution-gated multi-qubit \glspl{QApp} include distributed quantum computing~\cite{Main2025DQC,Barral2025DQC}, entanglement-assisted sensing and synchronization~\cite{IloOkeke2020QCS,Barhoumi2024QANT}, and \gls{QN} calibration and benchmarking~\cite{Pappa2012PRL,McCutcheon2016NatComm}.

\textbf{Motivation.}
Direct quantum teleportation over long distances is severely constrained by photon losses, which reduce both the achievable range and the probability of establishing the Bell pairs required for teleportation~\cite{pirandola2017fundamental}. \glspl{QR} alleviate this limitation by dividing a long quantum link into shorter segments and using entanglement swapping to distribute Bell pairs over longer distances~\cite{TelecomWavelengthQR}. This has motivated extensive work on \gls{QR}-assisted \glspl{QN}, including queueing-based models of Bell-pair generation and storage~\cite{Dai20IEEEJSAC}, continuous-time Markov models for entanglement switching~\cite{Vardoyan20ScienceDirectPerfEval,Vardoyan21IEEETransQEng,Vardoyan23ACMTMPE,Nain20ACMMACS}, and architectural performance studies~\cite{Panigrahy23INFOCOM,Tillman24IEEEQCE}. However, these works focus on link- and protocol-level behavior and do not address \gls{QApp}-level success conditions. In particular, they do not quantify \emph{reliability}, understood here as the probability that all qubits required by a \gls{QApp} are simultaneously available with fidelity above a target threshold after storage in quantum memories. The same limitation affects widely used \gls{QN} simulators~\cite{Bel2025SimulatorsReview}, such as QKDNetSim+, NetSquid, and QuISP, whose abstractions do not model continuous-time fidelity degradation during passive storage and therefore cannot assess the feasibility of multi-qubit teleportation under fidelity constraints.

\textbf{Contributions.}
This paper addresses the following question: given a \gls{QR}-assisted \gls{QN} with stochastic Bell-pair generation and quantum-memory decoherence, \emph{how many qubits can be teleported while all of them simultaneously satisfy a target fidelity constraint?} To answer this question, we study the scalability of fidelity-constrained multi-qubit teleportation in a two-node \gls{QN} assisted by a single \gls{QR}. The main contributions are:

\textbf{C1) \gls{QApp}-Level Reliability Metric:}
We formulate a \gls{QApp}-level \emph{reliability} metric that quantifies the probability that all Bell pairs required for a multi-qubit teleportation phase are simultaneously available with fidelity above a target threshold, jointly capturing stochastic entanglement generation, parallelism, and memory decoherence.

\textbf{C2) Reliability-Oriented Simulation Framework:}
We develop a Monte Carlo--based simulator that captures stochastic Bell-pair generation, parallel entanglement attempts, \gls{QR}-assisted entanglement distribution, and continuous-time fidelity degradation in quantum memories, enabling the evaluation of \gls{QApp}-level reliability.

\textbf{C3) Feasibility Analysis Across Technologies:}
We characterize the feasible operating regions in which a given number of qubits can be teleported while satisfying target fidelity and reliability constraints, and we assess how these regions depend on the quantum-memory technology, i.e., \gls{NV}-based centers or trapped ions, and on the transmission medium, i.e., optical fiber or terrestrial \gls{FSO}.

\textbf{Paper Outline.}
Section~\ref{sec:background} reviews the fundamentals of \gls{QN}. Section~\ref{sec:system_model} presents the system model. Section~\ref{sec:model} introduces the \gls{QApp}-level \emph{reliability} formulation. Section~\ref{sec:experimental_setup} describes the simulation setup. Section~\ref{sec:results} discusses the results, and Section~\ref{sec:conclusions} concludes the paper.

\section{Fundamentals}
\label{sec:background}

\subsection{Qubits and Quantum Entanglement}

A qubit is the fundamental unit of quantum information and can be expressed as a coherent superposition $\ket{\psi}=a\ket{0}+b\ket{1}$, with $a,b\in\mathbb{C}$ and $|a|^2+|b|^2=1$. In \glspl{QN}, qubits cannot be copied and must therefore be transferred between remote nodes using quantum teleportation. This protocol relies on bipartite quantum entanglement, typically in the form of Bell states, $\ket{\Phi^\pm}=\tfrac{1}{\sqrt{2}}(\ket{00}\pm\ket{11})$ and $\ket{\Psi^\pm}=\tfrac{1}{\sqrt{2}}(\ket{01}\pm\ket{10})$, which form an orthonormal basis of two-qubit states.

\subsection{Quantum Teleportation}
\label{sec:QuantumTeleportation}

Quantum teleportation enables the transfer of an unknown single-qubit state $\ket{\psi}$ from a source node (Alice) to a destination node (Bob). We consider teleportation assisted by an intermediate \gls{QR}, as illustrated in Fig.~\ref{fig:system-model}. First, the \gls{QR} probabilistically generates one Bell pair over the \gls{QR}--Alice link and another over the \gls{QR}--Bob link, storing one qubit of each pair locally while the other qubits are stored at Alice and Bob, respectively. Once both links succeed, the \gls{QR} performs a \gls{BSM} on its two stored qubits, thereby applying entanglement swapping and creating an end-to-end Bell pair shared between Alice and Bob. Let $A_1$ and $B$ denote the qubits of this resulting Bell pair stored at Alice and Bob, respectively, and let $A_0$ denote Alice's input qubit prepared in state $\ket{\psi}$. Alice then performs a joint \gls{BSM} on $A_0$ and $A_1$, which maps $\ket{\psi}$ onto Bob's qubit $B$ up to a Pauli byproduct operator. Finally, the outcomes of the \gls{BSM} operations at the \gls{QR} and at Alice are sent to Bob through a classical channel, allowing Bob to apply the appropriate Pauli correction and recover $\ket{\psi}$ deterministically.

Because entanglement generation, swapping, and classical signaling are not instantaneous, qubits may remain stored in the quantum memories of Alice, the \gls{QR}, and Bob while the protocol is completed. During this storage time, decoherence progressively degrades the corresponding quantum states over a characteristic time scale determined by the memory coherence time $\tau$.

\begin{figure}[t!]
  \centering
  \includegraphics[width=\columnwidth]{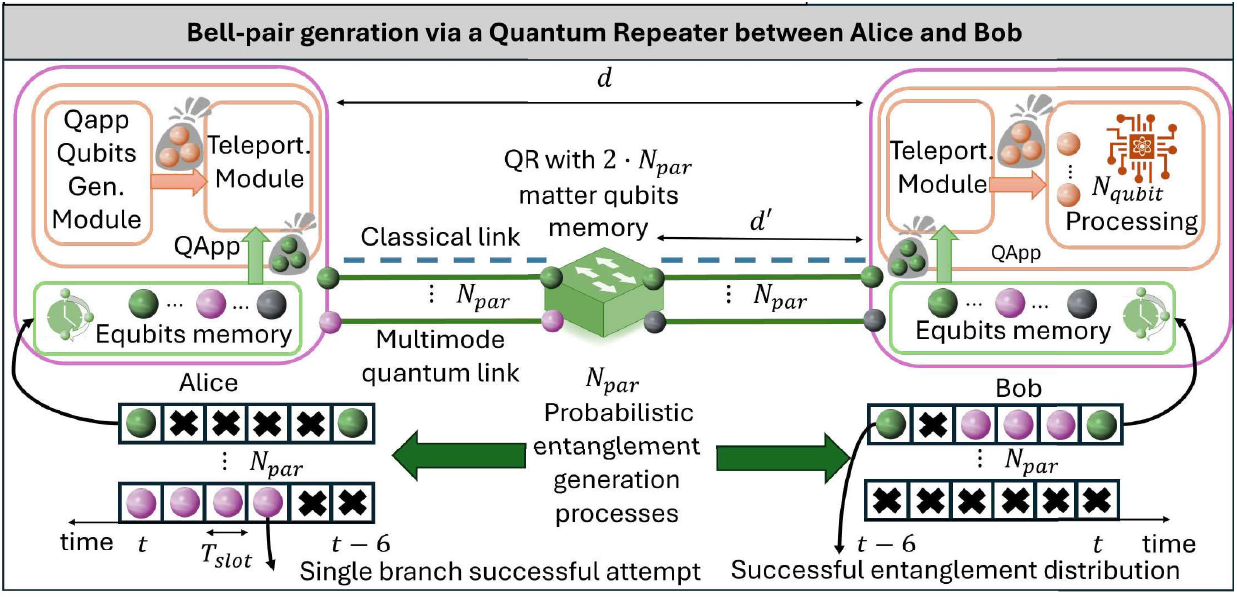}
  \caption{A \gls{QN} with two endpoints, connected via a \gls{QR}, illustrating parallel Bell pair distribution of qubits.}
  \label{fig:system-model}
\end{figure}

\section{System Model}
\label{sec:system_model}

We consider the \gls{QN} in Fig.~\ref{fig:system-model}, which consists of two remote nodes, Alice and Bob, connected through a single intermediate \gls{QR}. Alice and Bob host a distributed \gls{QApp} that requires the teleportation of $N_{\text{qubit}}>1$ independent qubits from Alice to Bob. Each qubit $n \in \{1,\ldots,N_{\text{qubit}}\}$ is teleported using a distinct end-to-end Bell pair obtained by generating one elementary Bell pair on the \gls{QR}--Alice link and one on the \gls{QR}--Bob link, followed by entanglement swapping at the \gls{QR}. The teleportation phase is successful only if every teleported qubit satisfies $F(\rho^{(n)},\psi^{(n)}) \geq F_{\mathrm{th}}$.

The \gls{QR} supports parallel entanglement-generation attempts with Alice and Bob. It is equipped with matter-based quantum memory of capacity $2N_{\text{par}}$ and coherence time $\tau_{\mathrm{QR}}$, where $N_{\text{par}}$ is the maximum number of parallel attempts. Successfully generated qubits are buffered at the \gls{QR} until both elementary links are available for entanglement swapping. Entanglement generation is probabilistic and proceeds in discrete time slots of duration
\begin{equation}
    T_{\text{slot}} = T_{\text{comm}} + T_{\text{att}} + T_{\text{BSM}} + T_{\text{PC}}.
\end{equation}
Here, $T_{\text{comm}} = 2d'/c$ is the round-trip classical signaling delay between the \gls{QR} and a remote node, where $d'$ is the \gls{QR}--remote-node distance. In the considered symmetric topology, the \gls{QR} is placed midway between Alice and Bob, so that $d'=d/2$, where $d$ is the Alice--Bob distance. The remaining terms denote the duration of an entanglement-generation attempt ($T_{\text{att}}$), the \gls{BSM} at the \gls{QR} and Alice ($T_{\text{BSM}}$), and the conditional Pauli corrections after classical signaling ($T_{\text{PC}}$). Entanglement generation and swapping follow a standard heralded protocol.

After entanglement swapping, the resulting end-to-end Bell pairs are stored in the local quantum memories of Alice and Bob, with capacities $N_{\text{qubit}}$ and coherence times $\tau_A$ and $\tau_B$, respectively. Each Bell pair is reserved for one qubit and decoheres independently while stored. Qubits are teleported individually as soon as their corresponding Bell pairs become available, and the teleported qubits are then stored at Bob, where decoherence starts upon teleportation. Let $\rho_B^{(q)}$ denote the state of qubit $q$ stored at Bob after teleportation. The \gls{QApp} execution is successful if, when the last qubit has been teleported, $F(\rho_B^{(q)},\psi^{(q)}) \geq F_{\mathrm{th}}$ holds for all $q \in \{1,\ldots,N_{\text{qubit}}\}$. Under the considered teleportation model, this condition can be equivalently assessed through the fidelity of the end-to-end Bell pairs used for teleportation. Otherwise, the execution fails and the whole procedure must be repeated with a new set of $N_{\text{qubit}}$ Bell pairs.

The model is restricted to a two-node \gls{QN} assisted by a single intermediate \gls{QR} and to raw entanglement distribution, i.e., without entanglement purification or distillation. This scope isolates the joint effect of stochastic entanglement-generation delay, parallelism, swapping quality, and quantum-memory decoherence on \gls{QApp}-level reliability. The general formulation keeps the coherence times at Alice, Bob, and the \gls{QR}, i.e., $\tau_A$, $\tau_B$, and $\tau_{\mathrm{QR}}$, explicit, while the experimental evaluation later specializes them to representative configurations.

\subsection{Quantum Channel Model}
\label{subsec:quantum_channel}

We consider two transmission media for entanglement distribution between the \gls{QR} and the remote nodes.

\textbf{Optical Fiber Links:}
Entanglement distribution is realized over optical fiber links using traveling photonic qubits. Each fiber supports one entanglement-generation attempt, so parallelism is achieved by deploying multiple fibers in parallel between the \gls{QR} and Alice or Bob. Attempts over different fibers are assumed independent and non-interfering. Hence, $N_{\text{par}}$ scales with the number of available fibers and is mainly limited by the deployed infrastructure.

\textbf{Terrestrial \gls{FSO} Links:}
Entanglement distribution is realized over line-of-sight \gls{FSO} links. Parallel entanglement-generation attempts are enabled by transmitting multiple spatially separated optical beams between the \gls{QR} and Alice or Bob. Since each beam occupies a distinct spatial mode and is detected independently, different attempts are distinguishable and non-interfering~\cite{Koudia2025SpatialModeDiversity}. In this case, the achievable parallelism $N_{\text{par}}$ is constrained by the spatial degrees of freedom of the \gls{FSO} channel, which depend on the channel space--bandwidth product and practical factors such as beam divergence, atmospheric turbulence, pointing accuracy, and receiver aperture size. For realistic terrestrial links, these constraints typically limit the number of independently addressable spatial modes to $1 \leq N_{\text{par}} \lesssim 10$~\cite{ZhaoFSOCapacity2015}.

\subsection{Entanglement Success Probability}
\label{subsec:entanglement_success_probability}

We model the effective transmission efficiency of the quantum links between the \gls{QR} and Alice or Bob as
\begin{equation}
\eta_\ell =
P_e \,
P_{\mathrm{det}} \,
P_{ce}^{(\ell)} \,
C_\ell \,
e^{-d_\ell^{\prime} / L_\ell},
\label{eq:eta_generic}
\end{equation}
where $\ell \in \{\text{fiber},\,\text{FSO}\}$ denotes the transmission medium and $d_\ell^{\prime}$ the corresponding physical link length. The parameter $P_e$ is the probability that a matter qubit emits a photon into the intended photonic mode, and $P_{\mathrm{det}}$ is the probability that a photon arriving at the receiver is detected. Both are intrinsic node-level efficiencies and are assumed independent of the transmission medium. The factor $P_{ce}^{(\ell)}$ accounts for coupling losses between the local photonic mode and the propagation mode of medium $\ell$, including interface losses and, for \gls{FSO}, residual effects such as pointing errors and turbulence-induced fading. Finally, $L_\ell$ is the attenuation length of medium $\ell$.

For optical fiber, i.e., $\ell=\text{fiber}$, propagation occurs in a guided medium and geometric coupling losses are neglected, so $C_{\text{fiber}} = 1$. For \gls{FSO}, i.e., $\ell=\text{FSO}$, propagation occurs over an unguided line-of-sight channel, and geometric coupling losses are modeled using a truncated Gaussian beam model~\cite{PirandolaFSO2021}. The corresponding coupling efficiency is
\begin{equation}
    C_{\text{FSO}}
=
1 - \exp\!\left[-\left(2 a_{\text{FSO}}^2\right)/\left(w^2\left[d_{\text{FSO}}^{\prime}\right]\right)\right],
\end{equation}
where $a_{\text{FSO}} = A_{\text{FSO}}/2$ is the receiver aperture radius and $w(d_{\text{FSO}}^{\prime})$ is the beam radius at distance $d_{\text{FSO}}^{\prime}$, given by
\begin{equation}
    w(d_{\text{FSO}}^{\prime}) =
w_0\sqrt{1 +\left(\lambda d_{\text{FSO}}^{\prime}/\left(\pi w_0^2\right)\right)^2},
\end{equation}
with $w_0$ the transmitter beam waist and $\lambda$ the optical wavelength. Given $\eta_\ell$, the probability of successfully generating a heralded entangled state over link $\ell$ in a single attempt is modeled as~\cite{Chehimi2025}
\begin{equation}
p_{\mathrm{suc}}^{(\ell)} =
2 \, \eta_\ell \cos^2\theta
\left( 1 - \eta_\ell^2 \cos^2\theta \right),
\label{eq:psuc_generic}
\end{equation}
where $\theta$ parameterizes the matter--photon entangled-state preparation.

\subsection{Bell-Pair Fidelity Evolution under Memory Decoherence}
\label{subsec:fidelity_evolution}

As discussed in Section~\ref{sec:QuantumTeleportation}, teleportation includes several storage phases in which the qubits forming a Bell pair may reside in the quantum memories of Alice, Bob, and the \gls{QR}, with coherence times $\tau_A$, $\tau_B$, and $\tau_{\mathrm{QR}}$, respectively. For any storage phase, we assume independent depolarizing noise on the two stored qubits. Under this assumption, the Bell-pair fidelity at time $t$ is
\begin{equation}
F(t,t_{1},t_{2}) =
\frac{3}{4}\left(\frac{4F_0-1}{3}\right)
e^{-\frac{t-t_{1}}{\tau_1}}
e^{-\frac{t-t_{2}}{\tau_2}}
+\frac{1}{4},
\label{eq:fidelity_time}
\end{equation}
where $t_1$ and $t_2$ are the time instants at which the two qubits start being stored, and $\tau_1$ and $\tau_2$ are the corresponding coherence times, instantiated as $\tau_{\mathrm{QR}}$, $\tau_A$, or $\tau_B$ depending on the hosting nodes. The parameter $F_0$ denotes the Bell-pair fidelity at the beginning of the considered storage phase. The exponential terms model independent coherence decay of the two qubits, while the affine form ensures convergence to $1/4$ for long storage times, corresponding to the overlap with a maximally mixed two-qubit state under depolarizing noise.

For the elementary Bell pairs generated on branches $A$ and $B$, the initial fidelities of their respective storage phases are denoted by $F_{0,A}$ and $F_{0,B}$. Accordingly, Eq.~\eqref{eq:fidelity_time} is applied to each elementary pair with $F_0=F_{0,A}$ or $F_0=F_{0,B}$, depending on the branch. Once both elementary links are available, entanglement swapping at the \gls{QR} creates an end-to-end Alice--Bob Bell pair. To describe this step compactly, we define the Werner parameter as $\xi=(4F-1)/3$~\cite{bruss}. Let $\xi_A$ and $\xi_B$ denote the Werner parameters of the two elementary Bell pairs at the instant of swapping, after any storage-induced decoherence accumulated before the \gls{BSM}. The Werner parameter of the post-swapping Bell pair is modeled as $\xi_0^{\mathrm{swap}}=\eta_{\mathrm{swap}}\,\xi_A \xi_B$, where $\eta_{\mathrm{swap}}\in(0,1]$ is a swapping-quality factor. Equivalently, the initial fidelity of the subsequent Alice--Bob storage phase is $F_0^{\mathrm{swap}}=\frac{1}{4}+\frac{3}{4}\eta_{\mathrm{swap}}\,\xi_A \xi_B$. Therefore, Eq.~\eqref{eq:fidelity_time} is applied with the appropriate phase-dependent initialization: $F_0=F_{0,A}$ or $F_0=F_{0,B}$ for elementary-link storage phases, and $F_0=F_0^{\mathrm{swap}}$ for the final storage phase of the end-to-end Bell pair after entanglement swapping.

\section{Reliability of Multi-Qubit Teleportation under Fidelity Constraints}
\label{sec:model}

\subsection{Reliability Definition}

Let $R$ denote the probability that the teleportation of $N_{\text{qubit}}$ qubits is completed while all teleported qubits satisfy the target fidelity constraint. Let $t_{\text{last}}$ be the time slot at which the last of the $N_{\text{qubit}}$ required Bell pairs has been successfully generated, distributed, and stored at Bob. The fidelity requirement is
$\min_{n \in \{1,\dots,N_{\text{qubit}}\}} F_n(t_{\text{last}}) \geq F_{\mathrm{th}}$,
where $F_n(t)$ denotes the fidelity of the Bell pair associated with qubit $n$ at time $t$. Since $t_{\text{last}}$ is random, the reliability is obtained by averaging over all possible completion times:
\begin{equation}
R =
\sum_{t_{\text{last}}=1}^{T_{\max}}
\mathbb{P}_{\mathrm{se}}[t_{\text{last}}]\,
R_{t_{\text{last}}}^{\mathrm{cond}},
\label{eq:expected_reliability}
\end{equation}
where $\mathbb{P}_{\mathrm{se}}[t_{\text{last}}]$ is the probability that all $N_{\text{qubit}}$ Bell pairs have been generated and stored by slot $t_{\text{last}}$, and $R_{t_{\text{last}}}^{\mathrm{cond}}$ is the conditional probability that the fidelity constraint is satisfied at that instant. In principle, the sum extends to infinity; in practice, $T_{\max}$ is chosen such that
$\sum_{t_{\text{last}}=1}^{T_{\max}} \mathbb{P}_{\mathrm{se}}[t_{\text{last}}] \ge 1-\epsilon$,
with $\epsilon$ a small tolerance (e.g., $10^{-6}$).

\subsection{Conditional Reliability at Completion Time}

The conditional reliability at completion time $t_{\text{last}}$ is $R_{t_{\text{last}}}^{\mathrm{cond}}
= \mathbb{P}\!\left( \mathbf{F}(t_{\text{last}}) \in \mathcal{F}_{\mathrm{th}} \right)$, where
$\mathbf{F}(t_{\text{last}})=\big(F_1(t_{\text{last}}),\ldots,F_{N_{\text{qubit}}}(t_{\text{last}})\big)$
collects the fidelities of all Bell pairs at $t_{\text{last}}$, and
$\mathcal{F}_{\mathrm{th}}=\{\mathbf{f}: f_n \geq F_{\mathrm{th}},\ \forall n\}$
is the feasible fidelity region. For tractability, we assume independent decoherence across qubits. This is reasonable when coherence is dominated by local noise sources, as commonly occurs in trapped-ion platforms~\cite{Bruzewicz2019} and solid-state quantum-network nodes such as NV centers~\cite{Bradley2019}. Under this assumption,
\begin{equation}
R_{t_{\text{last}}}^{\mathrm{cond}}
=
\prod_{n=1}^{N_{\text{qubit}}}
\mathbb{P}_{\mathrm{fid}}[n,t_{\text{last}}],
\label{eq:conditional_reliability}
\end{equation}
where $\mathbb{P}_{\mathrm{fid}}[n,t_{\text{last}}]$ is the probability that qubit $n$ satisfies the fidelity constraint at $t_{\text{last}}$.

\subsection{Fidelity Constraints at Completion Time}
\label{subsec:fidelity_prior_e2e}

With an intermediate \gls{QR}, entanglement generation for qubit $n$ involves two independent processes: one on the \gls{QR}--Alice link (branch $A$) and one on the \gls{QR}--Bob link (branch $B$). These processes succeed at random time slots $t_A$ and $t_B$ with probabilities $\mathbb{P}_{\mathrm{qubit}}^{A}[n,t_A]$ and $\mathbb{P}_{\mathrm{qubit}}^{B}[n,t_B]$, respectively. Since the two branches generally complete at different times, the earlier elementary Bell pair may decohere before entanglement swapping, and the resulting end-to-end Bell pair may then undergo an additional storage phase before $t_{\text{last}}$. The probability that qubit $n$ satisfies the fidelity constraint at completion time $t_{\text{last}}$ is obtained by averaging over all pairs $(t_A,t_B)$:
\begin{equation}
\mathbb{P}_{\mathrm{fid}}[n, t_{\text{last}}]
=
\frac{
\sum_{t_A,t_B=1}^{t_{\text{last}}}
\mathbb{P}_{\mathrm{qubit}}^{A}[n, t_A]\,
\mathbb{P}_{\mathrm{qubit}}^{B}[n, t_B]\,
\mathbf{1}_{(t_A,t_B)}
}{
\sum_{t_A,t_B=1}^{t_{\text{last}}}
\mathbb{P}_{\mathrm{qubit}}^{A}[n, t_A]\,
\mathbb{P}_{\mathrm{qubit}}^{B}[n, t_B]
}
\label{eq:Probability_fidelity_qr}
\end{equation}

where $\mathbf{1}_{(t_A,t_B)}$ equals $1$ if the fidelity constraint is satisfied for the pair $(t_A,t_B)$, and $0$ otherwise. The feasible pairs are characterized by the following lemma.

\begin{lemma}[Storage-Time Feasibility Condition]
\label{lem:storage_feasibility_qr}
Consider a qubit $n$ for which entanglement generation on branches $A$ and $B$ succeeds at time slots $(t_A,t_B)$. Let $\xi_{0,A}=(4F_{0,A}-1)/3$ and $\xi_{0,B}=(4F_{0,B}-1)/3$, and define $t'=|t_A-t_B|$ and $t''=t_{\text{last}}-\max\{t_A,t_B\}$. Let $(\tau',\tau'')=(\tau_A,\tau_B)$ if $t_A\le t_B$, and $(\tau',\tau'')=(\tau_B,\tau_A)$ otherwise, so that $\tau'$ corresponds to the endpoint that succeeds first. Then, the resulting end-to-end Bell pair satisfies $F(t_{\text{last}},t_A,t_B)\geq F_{\mathrm{th}}$ at completion time $t_{\text{last}}$ if and only if
\begin{equation}
\label{eq:storage_feasibility_qr_eta}
\begin{aligned}
t'\,\tau''(\tau'+\tau_{QR}) + t''\,\tau_{QR}(\tau'+\tau'')
\leq\;
\\
-\frac{\tau_{QR}\tau_A\tau_B}{T_{\text{slot}}}
\ln\!\left(\frac{4F_{\mathrm{th}}-1}{3\eta_{\mathrm{swap}}\xi_{0,A}\xi_{0,B}}\right),
\end{aligned}
\end{equation}
provided that $F_{\mathrm{th}}\leq \left(1+3\eta_{\mathrm{swap}}\xi_{0,A}\xi_{0,B}\right)/4$.
\end{lemma}

\begin{IEEEproof}
Under the depolarizing-noise model, let the elementary Bell pair generated first have coherence time $\tau'$ and initial Werner parameter $\xi_0'$, where $(\xi_0',\tau')=(\xi_{0,A},\tau_A)$ if $t_A\le t_B$, and $(\xi_{0,B},\tau_B)$ otherwise. Let $\xi_0''$ denote the initial Werner parameter of the later branch, i.e., $\xi_0''=\xi_{0,B}$ if $t_A\le t_B$, and $\xi_0''=\xi_{0,A}$ otherwise. After the first storage stage of duration $t'$, the earlier elementary Bell pair has parameter $\xi' = \xi_0' \exp\!\left[-T_{\text{slot}}t'\left(\frac{1}{\tau_{QR}}+\frac{1}{\tau'}\right)\right]$. At the instant of swapping, the later branch contributes $\xi_0''$, so the post-swapping end-to-end Bell pair starts the second storage stage with parameter $\xi_0^{\mathrm{swap}}=\eta_{\mathrm{swap}}\,\xi'\xi_0''=\eta_{\mathrm{swap}}\,\xi_{0,A}\xi_{0,B}\exp\!\left[-T_{\text{slot}}t'\left(\frac{1}{\tau_{QR}}+\frac{1}{\tau'}\right)\right]$. The second storage stage lasts $t''$ slots and involves Alice and Bob, yielding $\xi''=\eta_{\mathrm{swap}}\,\xi_{0,A}\xi_{0,B}\exp\!\left\{-T_{\text{slot}}\left[t'\left(\frac{1}{\tau_{QR}}+\frac{1}{\tau'}\right)+t''\left(\frac{1}{\tau_A}+\frac{1}{\tau_B}\right)\right]\right\}$. Hence, $F''=\frac{1}{4}+\frac{3}{4}\xi''$. Enforcing $F''\geq F_{\mathrm{th}}$, taking logarithms, and rearranging yields Eq.~\eqref{eq:storage_feasibility_qr_eta}. The additional condition follows from requiring $(4F_{\mathrm{th}}-1)/(3\eta_{\mathrm{swap}}\xi_{0,A}\xi_{0,B})\le 1$.
\end{IEEEproof}

In this work, the probabilities $\mathbb{P}_{\mathrm{se}}[t]$ and $\mathbb{P}_{\mathrm{qubit}}[n,t]$ are obtained numerically through Monte Carlo simulation.

\subsection{Monte Carlo Simulation Method}
\label{subsec:mc_simulator}

The Monte Carlo simulator models a single \gls{QR}-assisted teleportation execution stage between Alice and Bob. Time is discretized into slots of duration $T_{\text{slot}}$. In each run, entanglement-generation attempts are performed in parallel on the \gls{QR}--Alice and \gls{QR}--Bob links, with up to $N_{\text{par}}$ simultaneous attempts per slot. Each attempt succeeds or fails according to the stochastic entanglement-generation model in Section~\ref{subsec:entanglement_success_probability}. For each qubit, the simulator records the success times $(t_A,t_B)$ on the two branches and defines $t_{\text{last}}$ as the slot at which all $N_{\text{qubit}}$ Bell pairs have been successfully generated and swapped. Given $(t_A,t_B,t_{\text{last}})$, the qubit fidelities at completion time are evaluated according to the phase-dependent fidelity initialization described in Section~\ref{subsec:fidelity_evolution} and the feasibility condition of Lemma~\ref{lem:storage_feasibility_qr}. A run is declared successful if all $N_{\text{qubit}}$ qubits satisfy the target fidelity constraint $F_{\mathrm{th}}$ at completion time. Repeating this procedure over many independent runs yields empirical estimates of $\mathbb{P}_{\mathrm{qubit}}[n,t]$, $\mathbb{P}_{\mathrm{se}}[t]$, and the resulting \textit{reliability} $R$.

To facilitate reproducibility, the Python implementation of the Monte Carlo simulator is available at \url{https://github.com/oadamuz/How_many_qubits_can_be_teleported}.

\section{Experimental Setup}
\label{sec:experimental_setup}

\subsection{Parameter Configuration}
\label{subsec:parameter_configuration}

All experiments consider a single teleportation-based \gls{QApp} executed between Alice and Bob with one intermediate \gls{QR}, placed symmetrically between the endpoints, i.e., $d'=d/2$.

The constant parameters are listed in Table~\ref{tab:fixed_parameters}. Unless otherwise stated, the experiments assume ideal elementary-pair initialization and ideal entanglement swapping, i.e., $F_{0,A}=F_{0,B}=1$ and $\eta_{\mathrm{swap}}=1$, so the reported results isolate the impact of entanglement-generation delays and quantum-memory decoherence. We further assume $\tau_A = \tau_B = \tau_{\mathrm{QR}} \equiv \tau$, and representative coherence times are summarized in Table~\ref{tab:coherence_times}. The study dimensions are defined by varying $N_{\text{par}} \in [1,10]$, $F_{\mathrm{th}} \in [0.75,0.95]$, and $R_{\mathrm{th}} \in [0.8,0.95]$. For each configuration, $N_{\text{qubit}}$ is progressively increased, and the maximum feasible value is the largest $N_{\text{qubit}}$ satisfying $R \geq R_{\mathrm{th}}$.

\begin{table}[b!]
\centering
\caption{Parameters used in the experimental evaluation.}
\label{tab:fixed_parameters}
\small
\begin{tabular}{l c l c}
\hline
\textbf{Parameter} & \textbf{Value} & \textbf{Parameter} & \textbf{Value} \\
\hline
$\theta$ & $\pi/4$~\cite{Chehimi2025} &
$P_e$ & $0.6$~\cite{Chehimi2025} \\

$P_{\mathrm{det}}$ & $0.85$~\cite{Chehimi2025} &
$P_{ce}^{(\text{fiber})}$ & $0.9$~\cite{Chehimi2025} \\

$P_{ce}^{(\text{FSO})}$ & $0.8$~\cite{El-Wakeel:16} &
$T_{\text{att}}$ & $5.5~\mu\mathrm{s}$~\cite{Chehimi2025} \\

$T_{\text{BSM}}$ & $50~\mathrm{ns}$~\cite{Chehimi2025} &
$T_{\text{PC}}$ & $10~\mathrm{ns}$~\cite{Chehimi2025} \\

$L_{\text{fiber}}$ & $21.7~\mathrm{km}$~\cite{Morana2020UltraLowLossFiber} &
$L_{\text{FSO}}$ & $8.7~\mathrm{km}$~\cite{kolka2009simulation} \\

$A_{\text{FSO}}$ & $0.20~\mathrm{m}$~\cite{Singh2023} &
$\lambda$ & $1550~\mathrm{nm}$~\cite{Kaushal2017} \\

$w_0$ & $1.5~\mathrm{cm}$~\cite{Lushnikov2018Diffraction} &
$F_{0,A}$ & $1$ \\

$F_{0,B}$ & $1$ &
$\xi_{0,A}$ & $1$ \\

$\xi_{0,B}$ & $1$ &
$\eta_{\mathrm{swap}}$ & $1$ \\
\hline
\end{tabular}
\end{table}

\begin{table}[b!]
\footnotesize
\centering
\caption{Coherence times $\tau$ for NV-center and trapped-ion quantum memory platforms.}
\label{tab:coherence_times}
\begin{tabularx}{\columnwidth}{l X r}
\hline
\textbf{Platform} & \textbf{Configuration / Conditions} & \textbf{Coh. Times ($\tau$)} \\
\hline
\textbf{NV-Center} & High-purity, low $^{13}$C, cryogenic & $1$--$3$ ms \cite{Balasubramanian2009} \\
(Diamond) & High-purity, room temperature & $300$--$600$ $\mu$s \cite{Stanwix2010} \\
& Near-surface, room temperature & $10$--$100$ $\mu$s \cite{Sangtawesin2019} \\
& High nitrogen concentration & $<100$ $\mu$s \cite{Luo_2022} \\
\hline
\textbf{Trapped-Ion} & \textit{Optical}: passive memory$^a$ & $100$--$500$ ms \cite{Hempel2026TIQC} \\
\cline{2-3}
\textbf{Ground-state} & Passive, no active suppression & $\sim 10$--$60$ s \cite{PhysRevLett.113.220501} \\
(Hyperfine / & With DD and SC$^b$ & $>10$ min \cite{Wang2017TenMinute} \\
Clock) & Active noise stabilization & $\sim 1.5$ h \cite{Wang2021OneHour} \\
\hline
\multicolumn{3}{l}{\scriptsize $^a$Limited by laser stability. $^b$DD: Dynamical Decoupling; SC: Sympathetic Cooling.}
\end{tabularx}
\end{table}

\subsection{Experiments}
\label{subsec:experiments}

The evaluation is organized into three experiments. Experiments~1 and~2 provide a parametric characterization of \gls{QApp}-level reliability by analyzing the impact of parallel entanglement generation, end-to-end distance, and memory coherence time. To isolate these effects, both experiments consider fiber links and \gls{NV}-center-based memories. Experiment~3 addresses the central question of this work, namely how many qubits can be teleported reliably under fidelity constraints in different deployment scenarios, by extending the analysis to heterogeneous settings including \gls{FSO} links and trapped-ion-based memories, while remaining within the single-\gls{QR} scope of Section~\ref{sec:system_model}.

\textbf{Experiment 1: Reliability of Multi-Qubit Entanglement with Parallel Attempts.}
We generate \textit{reliability} curves for teleporting $N_{\text{qubit}} \in \{2,4,8\}$ qubits under fidelity constraints. The fidelity threshold is swept over $F_{\text{th}} \in [0.75,1]$, the number of parallel entanglement-generation attempts is set to $N_{\text{par}} \in \{1,4\}$, and the end-to-end distance takes values $d \in \{0.25,0.5,0.75,1.00,1.25\}$~km. The memory coherence time is fixed to $\tau = 500~\mu\mathrm{s}$.

\begin{figure*}[t!]
    \centering
    \includegraphics[width=\textwidth]{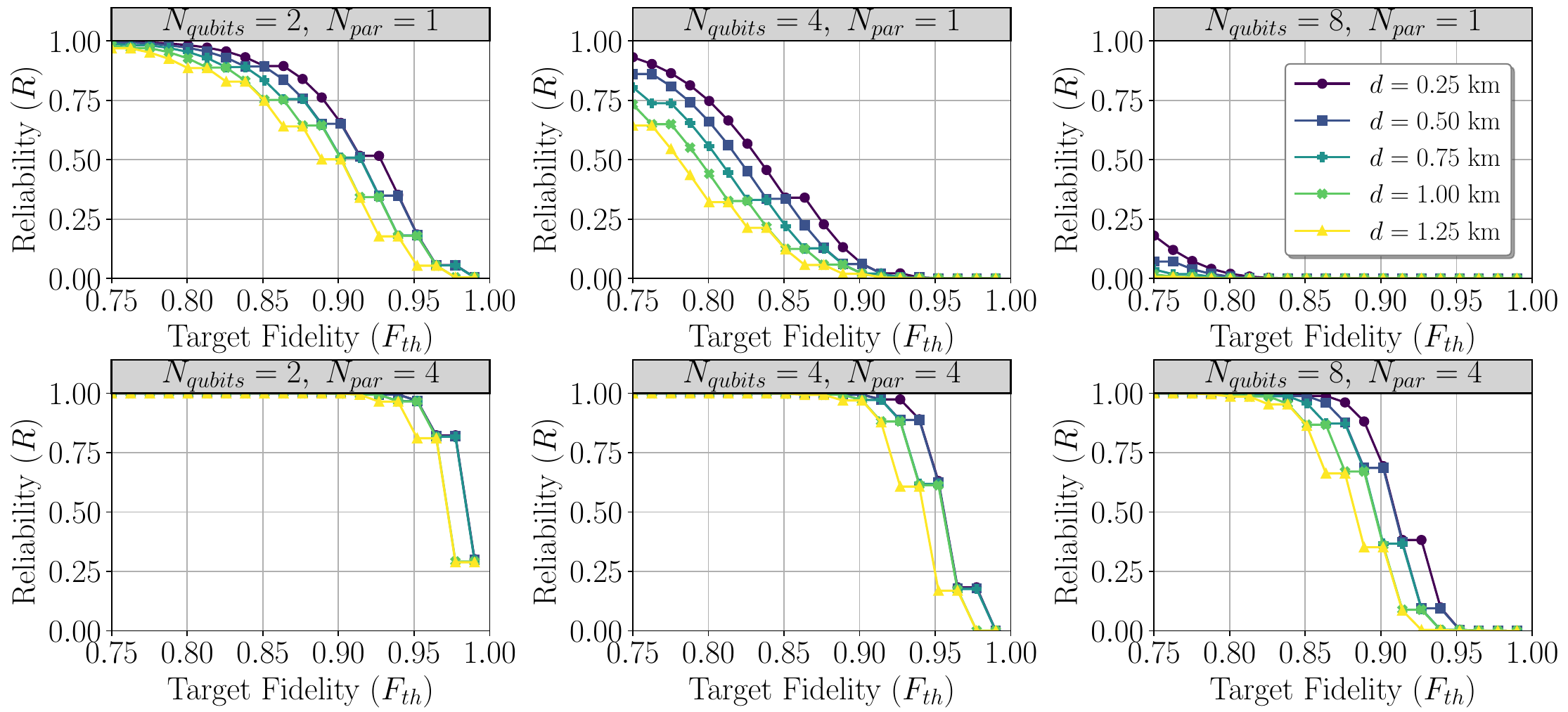}
    \caption{Evaluation of \textit{reliability} $R$ as a function of the target fidelity $F_{\text{th}}$ for $\tau = 500~\mu\mathrm{s}$.}
    \label{fig:reliability_vs_fidelity}
\end{figure*}

\textbf{Experiment 2: Impact of Memory Coherence Time on Reliability.}
We generate \textit{reliability} curves as a function of the memory coherence time $\tau$. The Alice--Bob distance is fixed to $d = 1.5~\mathrm{km}$, and the coherence time is swept over $\tau \in [10~\mu\mathrm{s},\,1~\mathrm{ms}]$. For each value of $\tau$, \textit{reliability} is evaluated for $F_{\text{th}} \in \{0.76,0.8,0.84,0.88,0.92,0.96\}$, $N_{\text{qubit}} \in \{2,4,8\}$, and $N_{\text{par}} \in \{1,4\}$.

\textbf{Experiment 3: Maximum Number of Teleportable Qubits.}
We determine the maximum number of teleportable qubits $N_{\text{qubit}}^{\max}$ that can be supported under a target \textit{reliability} $R_{\mathrm{th}}$ as a function of the Alice--Bob distance $d$. The analysis considers fiber and \gls{FSO} links, together with coherence-time regimes representative of \gls{NV}-center-based and trapped-ion-based platforms. For each combination of link type and memory technology, feasibility regions in the $(d, N_{\text{qubit}})$ plane are identified. The analysis considers $N_{\text{par}} \in \{1,2,4\}$ and $F_{\mathrm{th}} \in \{0.75,0.8,0.85\}$.

All experiments were executed on a machine equipped with 16\,GB RAM and a quad-core Intel Core i7-7700HQ processor running at 2.80\,GHz.

\section{Performance Results}
\label{sec:results} 

\subsection{Experiment 1: \textit{Reliability} of Multi-Qubit Entanglement with Parallel Attempts.}

Fig.~\ref{fig:reliability_vs_fidelity} illustrates the \textit{reliability} $R$ as a function of the target fidelity $F_{\mathrm{th}}$ for different combinations of $N_{\text{qubit}}$ and $N_{\text{par}}$.

First, increasing the number of qubits required by the \gls{QApp} significantly reduces the \textit{reliability}. This is due to the joint nature of the success event over all $N_{\text{qubit}}$ Bell pairs and to the longer entanglement-completion times, which amplify storage-induced decoherence. This effect is evident when comparing configurations with different values of $N_{\text{qubit}}$. For example, for $N_{\text{qubit}}=2$ and $N_{\mathrm{par}}=1$, the \textit{reliability} at $F_{\mathrm{th}}=0.75$ is $R\approx 1$ for $d=250~\mathrm{m}$ and $R=0.98$ for $d=500~\mathrm{m}$. In contrast, for $N_{\text{qubit}}=8$ and $N_{\mathrm{par}}=1$, the \textit{reliability} drops to $R=0.2$ at $d=250~\mathrm{m}$ and to $R=0.1$ at $d=500~\mathrm{m}$.

Second, enabling parallel entanglement-generation attempts has a substantial positive impact on system performance. This is evident when comparing the top row ($N_{\text{par}} = 1$) with the bottom row ($N_{\text{par}} = 4$). Parallelism shortens the entanglement phase and limits the accumulated decoherence. In the previous example with $N_{\text{qubit}} = 8$, increasing $N_{\text{par}}$ from $1$ to $4$ allows the system to achieve unit \textit{reliability} for a target fidelity of $F_{\mathrm{th}}=0.85$ when the distance is below $d=500~\mathrm{m}$.

\subsection{Experiment 2: Impact of Memory Coherence Time on Reliability.}

\begin{figure*}[t]
    \centering
    \includegraphics[width=\textwidth]{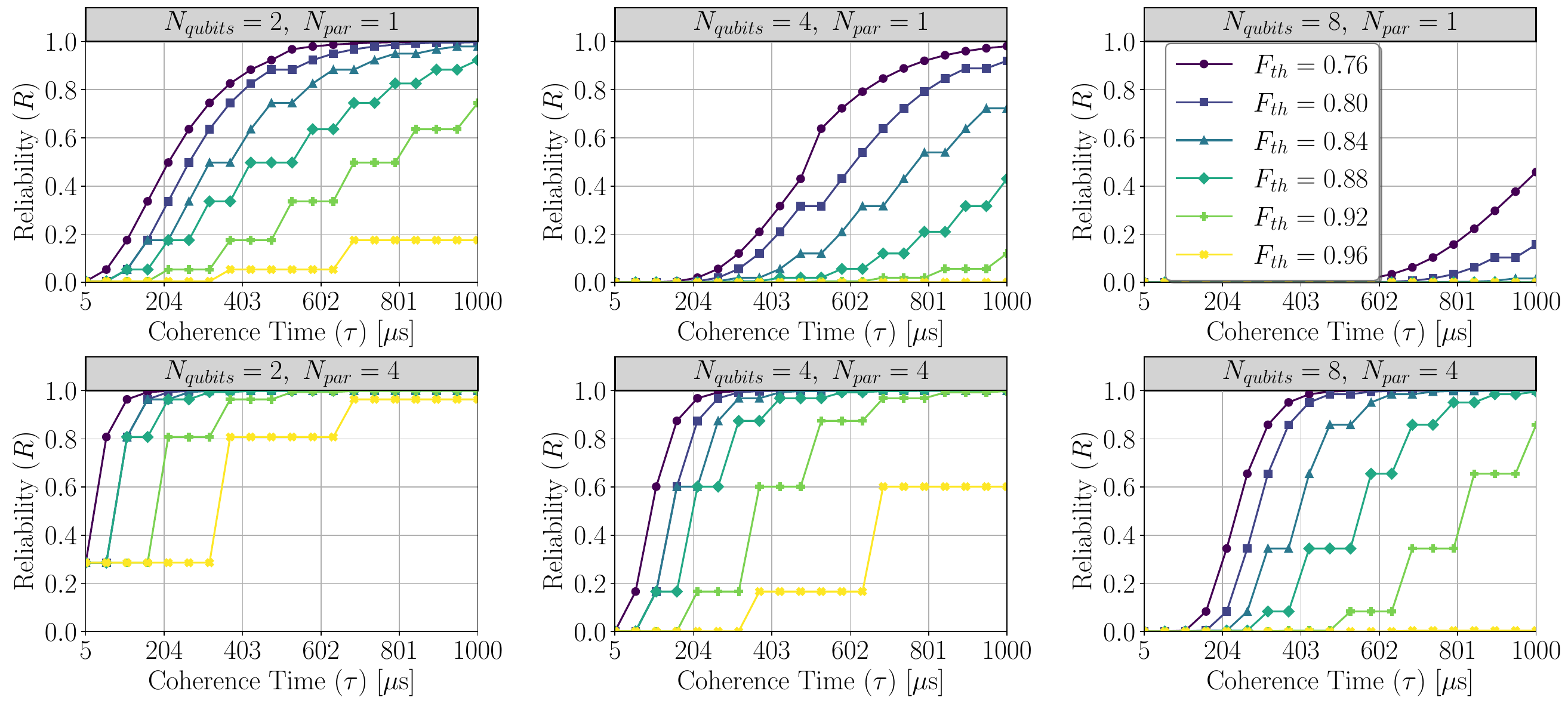}
    \caption{Evaluation of \textit{reliability} $R$ with respect to the coherence time $\tau$. Note that $d=1.5$ Km.}
    \label{fig:reliability_vs_coherence}
\end{figure*}
\begin{figure*}[t!]
    \centering
    \includegraphics[width=\textwidth]{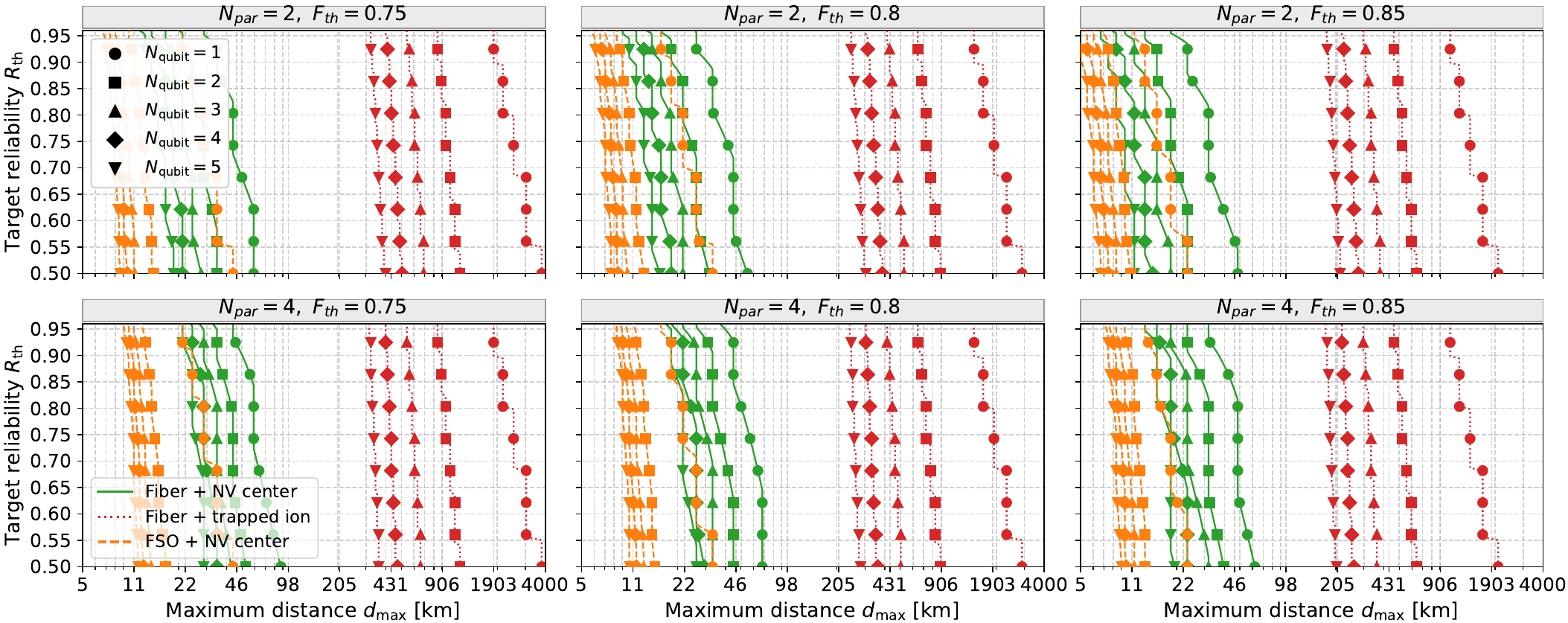}
    \caption{Maximum achievable distance $d_{\max}$ as a function of the target reliability $R_{\mathrm{th}}$ for different values of $N_{\text{qubit}}$, considering NV-center-based memories with coherence time $\tau = 3~\mathrm{ms}$ and trapped-ion-based memories with coherence time $\tau = 250~\mathrm{ms}$.}
    \label{fig:dmax_vs_rth}
\end{figure*}

Fig.~\ref{fig:reliability_vs_coherence} shows the resulting \textit{reliability} $R$ as a function of the memory coherence time $\tau$, for different combinations of $N_{\text{qubit}} \in \{2,4,8\}$, $N_{\text{par}} \in \{1,4\}$, and target fidelity values $F_{\text{th}}$.

The results show a clear trade-off between the number of qubits required by the \gls{QApp} and the achievable parallelism in the Bell-pair generation process. Increasing $N_{\text{qubit}}$ systematically degrades the \textit{reliability} $R$, while higher parallelism $N_{\text{par}}$ mitigates this effect by reducing the effective accumulation time. This follows from the joint success requirement over all $N_{\text{qubit}}$ Bell pairs and from the reduction of the completion time $t_{\text{last}}$ when parallel entanglement attempts are enabled. For instance, for $N_{\text{qubit}} = 8$ and $F_{\text{th}} = 0.76$, the reliability is $R\approx0.45$ for $N_{\text{par}} = 1$ when $\tau \approx 1~\mathrm{ms}$, whereas increasing the parallelism to $N_{\text{par}} = 4$ raises the reliability to $R \approx 1$. In contrast, for short coherence times ($\tau < 100~\mu\mathrm{s}$), \glspl{QApp} with $N_{\text{qubit}} = 8$ exhibit very low \textit{reliability} even when parallel entanglement attempts are enabled.

\subsection{Experiment 3: Maximum Number of Teleportable Qubits}

Fig.~\ref{fig:dmax_vs_rth} shows the maximum admissible Alice--Bob distance $d_{\max}$ as a function of the target reliability $R_{\mathrm{th}}$ for different values of $N_{\text{qubit}}$, $N_{\text{par}}$, $F_{\mathrm{th}}$, transmission media, and memory technology. Each subplot corresponds to a specific pair $(N_{\text{par}},F_{\mathrm{th}})$: the top row considers $N_{\text{par}}=2$, the bottom row $N_{\text{par}}=4$, and the three columns correspond to $F_{\mathrm{th}}\in\{0.75,0.8,0.85\}$. Within each subplot, each curve represents a fixed value of $N_{\text{qubit}}$, and its points define the feasibility boundary in the $(d_{\max},R_{\mathrm{th}})$ plane for a given transmission-medium and memory-technology combination.

For a fixed memory technology, transmission medium, and target reliability $R_{\mathrm{th}}$, increasing the number of teleported qubits $N_{\text{qubit}}$ shifts the feasible operating region toward shorter distances. This is observed as a leftward displacement of the feasibility boundary as $N_{\text{qubit}}$ increases, indicating a progressive reduction of $d_{\max}$. This trend is consistent across all considered values of $N_{\text{par}}$ and $F_{\mathrm{th}}$.

The figure also shows the impact of the fidelity target and the degree of parallelism. For a fixed deployment setting and $N_{\text{qubit}}$, moving from left to right, i.e., increasing $F_{\mathrm{th}}$ from $0.75$ to $0.85$, reduces the maximum feasible distance. In contrast, increasing the number of parallel entanglement-generation attempts from $N_{\text{par}}=2$ to $N_{\text{par}}=4$ shifts the feasibility boundaries toward larger values of $d_{\max}$. Hence, parallelism partially compensates for the stricter storage-time requirements imposed by higher fidelity targets, although this gain becomes more limited as $N_{\text{qubit}}$ increases.

Comparing transmission media, optical fiber supports larger values of $d_{\max}$ than terrestrial \gls{FSO} for the same $(N_{\text{qubit}},R_{\mathrm{th}})$ pair, as shown by the fiber-based curves appearing systematically to the right of the corresponding \gls{FSO} curves in the \gls{NV}-center regime. This is caused by the strong geometric coupling losses inherent to ground-based \gls{FSO} links, which limit the achievable distance even under clear-sky conditions. In the \gls{NV}-center regime, \gls{FSO} links yield feasible distances on the order of a few tens of kilometers, whereas fiber-based links can sustain distances of up to about $100~\mathrm{km}$ depending on the reliability and fidelity targets.

Memory technology has a dominant impact on the attainable distance scale. While \gls{NV}-center-based memories are limited to short- and medium-range links, trapped-ion-based memories substantially expand the feasible regions. For fiber links, trapped-ion-based memories allow $d_{\max}$ values of several hundreds of kilometers and, for $N_{\text{qubit}}=1$, even beyond $1000~\mathrm{km}$. For terrestrial \gls{FSO} links, feasible distances are inherently constrained by line-of-sight conditions and geometric coupling losses; in realistic ground-based scenarios, this typically limits operation to distances on the order of a few tens of kilometers even under favorable conditions. For this reason, \gls{FSO} configurations with trapped-ion-based memories are not evaluated.

\section{Conclusion}
\label{sec:conclusions} 
This paper studied the scalability of fidelity-constrained multi-qubit quantum teleportation between Alice and Bob assisted by a single \gls{QR}. We introduced a \gls{QApp}-level \textit{reliability} metric capturing the joint effect of stochastic end-to-end Bell-pair generation, parallel entanglement attempts, and quantum-memory decoherence, and evaluated it through Monte Carlo simulation. Results show that memory coherence time is the main scalability bottleneck, while parallelism is essential to sustain \textit{reliability} as the number of teleported qubits increases. Feasibility strongly depends on the underlying technology: \gls{NV}-center-based memories restrict operation to short and medium distances, whereas trapped-ion-based memories enable reliable multi-qubit teleportation over hundreds of kilometers in fiber links. In contrast, terrestrial \gls{FSO} links remain constrained by geometric coupling losses.

Future work will pursue two directions: deriving an analytical approximation of the proposed reliability metric to avoid computationally intensive Monte Carlo simulations, and extending the model to more general quantum-network settings, including multi-repeater topologies and entanglement purification/distillation mechanisms.


\section*{Acknowledgment}
This work is part of the project PID2022-137329OB-C43 funded by MICIU/AEI/10.13039/501100011033 and by FEDER, EU, and also part of the project C-ING-306-UGR23 funded by Consejería de Universidad, Investigación e Innovación and by ERDF Andalusia Program 2021-2027.



%



\bibliographystyle{ieeetr}
\bibliography{references}

@ARTICLE{Chehimi2025,
  author={Chehimi, Mahdi and Goodenough, Kenneth and Saad, Walid and Towsley, Don and Zhou, Tony X.},
journal={IEEE J. Sel. Areas Commun.
}, 
  title={{Entanglement Distribution Delay Optimization in Quantum Networks With Distillation}}, 
  year={2025},
  volume={43},
  number={5},
  pages={1871-1886},
  doi={10.1109/JSAC.2025.3543485}}

@article{Balasubramanian2009,
author       = {G. Balasubramanian and others},
  title        = {{Ultralong Spin Coherence Time in Isotopically Engineered Diamond}},
  journal      = {Nature Mater.},
  volume       = {8},
  pages        = {383--387},
  year         = {2009},
  doi          = {10.1038/nmat2420}
}

@article{Sangtawesin2019,
  title = {{Origins of Diamond Surface Noise Probed by Correlating Single-Spin Measurements with Surface Spectroscopy}},
author = {Sangtawesin, Sorawis and others},
  journal = {Phys. Rev. X},
  volume = {9},
  issue = {3},
  pages = {031052},
  numpages = {17},
  year = {2019},
  month = {Sep},
  publisher = {American Physical Society},
  doi = {10.1103/PhysRevX.9.031052},
  url = {https://link.aps.org/doi/10.1103/PhysRevX.9.031052}
}

@article{Luo_2022,
doi = {10.1088/1367-2630/ac58b6},
url = {https://dx.doi.org/10.1088/1367-2630/ac58b6},
year = {2022},
month = {mar},
publisher = {IOP Publishing},
volume = {24},
number = {3},
pages = {033030},
author = {Luo, T and others},
title = {{Creation of nitrogen-vacancy centers in chemical vapor deposition diamond for sensing applications}},
journal = {New Journal of Physics}
}

@article{Stanwix2010,
  title = {{Coherence of nitrogen-vacancy electronic spin ensembles in diamond}},
author = {Stanwix, P. L. and others},
  journal = {Phys. Rev. B},
  volume = {82},
  issue = {20},
  pages = {201201},
  numpages = {4},
  year = {2010},
  month = {Nov},
  publisher = {American Physical Society},
  doi = {10.1103/PhysRevB.82.201201},
  url = {https://link.aps.org/doi/10.1103/PhysRevB.82.201201}
}

@article{Wehner_QuantumInternet18,
author = {Stephanie Wehner  and David Elkouss  and Ronald Hanson },
title = {{Quantum internet: A vision for the road ahead}},
journal = {Science},
volume = {362},
number = {6412},
pages = {eaam9288},
year = {2018},
doi = {10.1126/science.aam9288},
URL = {https://www.science.org/doi/abs/10.1126/science.aam9288},
eprint = {https://www.science.org/doi/pdf/10.1126/science.aam9288},
}

@article{pirandola2017fundamental,
  title={{Fundamental limits of repeaterless quantum communications}},
  author={Pirandola, Stefano and Laurenza, Riccardo and Ottaviani, Carlo and Banchi, Leonardo},
  journal={Nat. Commun.},
  volume={8},
  number={1},
  pages={15043},
  year={2017},
  publisher={Nature Publishing Group UK London}
}

@ARTICLE{davies24tools,
  author={Davies, Bethany and Beauchamp, Thomas and Vardoyan, Gayane and Wehner, Stephanie},
journal={IEEE Trans. on Quantum Eng.}, 
  title={{Tools for the Analysis of Quantum Protocols Requiring State Generation Within a Time Window}}, 
  year={2024},
  volume={5},
  number={},
  pages={1-20},
  doi={10.1109/TQE.2024.3358674}}

@article{Vardoyan23ACMTMPE,
author = {Vardoyan, Gayane and Nain, Philippe and Guha, Saikat and Towsley, Don},
title = {{On the Capacity Region of Bipartite and Tripartite Entanglement Switching}},
year = {2023},
issue_date = {June 2023},
publisher = {Association for Computing Machinery},
address = {New York, NY, USA},
volume = {8},
number = {1–2},
issn = {2376-3639},
url = {https://doi.org/10.1145/3571809},
doi = {10.1145/3571809},
journal = {ACM Trans. Model. Perform. Eval. Comput. Syst.},
month = mar,
articleno = {1},
numpages = {18}
}

@article{Vardoyan20ScienceDirectPerfEval,
title = {{On the Exact Analysis of an Idealized Quantum Switch}},
journal = {Perform. Eval.},
volume = {144},
pages = {102141},
year = {2020},
issn = {0166-5316},
doi = {https://doi.org/10.1016/j.peva.2020.102141},
url = {https://www.sciencedirect.com/science/article/pii/S0166531620300614},
author = {Gayane Vardoyan and Saikat Guha and Philippe Nain and Don Towsley}
}

@ARTICLE{Vardoyan21IEEETransQEng,
  author={Vardoyan, Gayane and Guha, Saikat and Nain, Philippe and Towsley, Don},
journal={IEEE Trans. on Quantum Eng.}, 
  title={{On the Stochastic Analysis of a Quantum Entanglement Distribution Switch}}, 
  year={2021},
  volume={2},
  number={},
  pages={1-16},
  doi={10.1109/TQE.2021.3058058}}

@ARTICLE{Dai20IEEEJSAC,
  author={Dai, Wenhan and Peng, Tianyi and Win, Moe Z.},
journal={IEEE J. Sel. Areas Commun.}, 
  title={{Quantum Queuing Delay}}, 
  year={2020},
  volume={38},
  number={3},
  pages={605-618},
  doi={10.1109/JSAC.2020.2969000}}

@INPROCEEDINGS{Panigrahy23INFOCOM,
  author={Panigrahy, Nitish K. and Vasantam, Thirupathaiah and Towsley, Don and Tassiulas, Leandros},
booktitle={IEEE INFOCOM}, 
  title={{On the Capacity Region of a Quantum Switch with Entanglement Purification}}, 
  year={2023},
  volume={},
  number={},
  pages={1-10},
  doi={10.1109/INFOCOM53939.2023.10229003}}

@INPROCEEDINGS{Tillman24IEEEQCE,
  author={Tillman, Ian and Vasantam, Thirupathaiah and Towsley, Don and Seshadreesan, Kaushik P.},
booktitle={IEEE QCE}, 
  title={{Calculating the Capacity Region of a Quantum Switch}}, 
  year={2024},
  volume={01},
  number={},
  pages={1868-1878},
  doi={10.1109/QCE60285.2024.00216}}

@article{Nain20ACMMACS,
author = {Nain, Philippe and Vardoyan, Gayane and Guha, Saikat and Towsley, Don},
title = {{On the Analysis of a Multipartite Entanglement Distribution Switch}},
year = {2020},
issue_date = {June 2020},
publisher = {Association for Computing Machinery},
address = {New York, NY, USA},
volume = {4},
number = {2},
url = {https://doi.org/10.1145/3392141},
doi = {10.1145/3392141},
journal = {Proc. ACM Meas. Anal. Comput. Syst.},
month = jun,
articleno = {23},
numpages = {39}
}

@article{TelecomWavelengthQR,
  title = {{Telecom-Wavelength Quantum Repeater Node Based on a Trapped-Ion Processor}},
  author = {Krutyanskiy, V. and others},
  journal = {Phys. Rev. Lett.},
  volume = {130},
  issue = {21},
  pages = {213601},
  numpages = {7},
  year = {2023},
  month = {May},
  publisher = {American Physical Society},
  doi = {10.1103/PhysRevLett.130.213601},
  url = {https://link.aps.org/doi/10.1103/PhysRevLett.130.213601}
}

@article{Main2025DQC,
  author  = {Main, D. and others},
  title   = {{Distributed quantum computing across an optical network link}},
  journal = {Nature},
  volume  = {638},
  pages   = {383--388},
  year    = {2025},
  doi     = {10.1038/s41586-024-08404-x}
}

@article{Barral2025DQC,
  author  = {Barral, David and others},
  title   = {{Review of Distributed Quantum Computing: From single QPU to High Performance Quantum Computing}},
  journal = {Comput. Sci. Rev.},
  volume  = {57},
  pages   = {100747},
  year    = {2025},
  doi     = {10.1016/j.cosrev.2025.100747}
}

@inproceedings{IloOkeke2020QCS,
  author    = {Ilo-Okeke, Ebubechukwu O. and Tessler, Louis and Dowling, Jonathan P. and Byrnes, Tim},
  title     = {{Entanglement-based quantum clock synchronization}},
  booktitle = {AIP Conf. Proc.},
  volume    = {2241},
  number    = {1},
  pages     = {020011},
  year      = {2020},
  doi       = {10.1063/5.0011396},
  url       = {https://pubs.aip.org/aip/acp/article/2241/1/020011/1002056/Entanglement-based-quantum-clock-synchronization}
}

@article{Barhoumi2024QANT,
  author  = {Barhoumi, Mohamed},
  title   = {{Quantum-assisted network time synchronisation: A literature review, considering examples and challenges}},
  journal = {SSRN Electron. J.},
  year    = {2024},
  url     = {https://papers.ssrn.com/sol3/papers.cfm?abstract_id=5170022}
}

@article{Pappa2012PRL,
  author  = {Pappa, Anna and Chailloux, Andr{\'e} and Wehner, Stephanie and Diamanti, Eleni and Kerenidis, Iordanis},
  title   = {{Multipartite Entanglement Verification Resistant against Dishonest Parties}},
  journal = {Phys. Rev. Lett.},
  volume  = {108},
  number  = {26},
  pages   = {260502},
  year    = {2012},
  doi     = {10.1103/PhysRevLett.108.260502}
}

@article{McCutcheon2016NatComm,
  author  = {McCutcheon, W. and others},
  title   = {{Experimental Verification of Multipartite Entanglement in Quantum Networks}},
  journal = {Nat. Commun.},
  volume  = {7},
  pages   = {13251},
  year    = {2016},
  doi     = {10.1038/ncomms13251}
}

@incollection{Hempel2026TIQC,
  author    = {Cornelius Hempel},
  title     = {Trapped-ion quantum computers},
  booktitle = {Quantum Technologies. Trends and Implications for Cyber Defense},
  editor    = {J. Jang-Jaccard and P. Caroff and E. Blezinger and V. Mulder and A. Mulder and V. Lenders},
  publisher = {Springer Nature},
  address   = {Cham},
  year      = {2026},
  pages     = {15--24},
  doi       = {10.1007/978-3-031-90727-2}
}

@article{PhysRevLett.113.220501,
  title = {High-Fidelity Preparation, Gates, Memory, and Readout of a Trapped-Ion Quantum Bit},
  author = {Harty, T. P. and others},
  journal = {Phys. Rev. Lett.},
  volume = {113},
  issue = {22},
  pages = {220501},
  numpages = {5},
  year = {2014},
  month = {Nov},
  publisher = {American Physical Society},
  doi = {10.1103/PhysRevLett.113.220501},
  url = {https://link.aps.org/doi/10.1103/PhysRevLett.113.220501}
}

@article{Wang2017TenMinute,
  author  = {Wang, Ye and others},
  title   = {{Single-qubit quantum memory exceeding ten-minute coherence time}},
  journal = {Nat. Photon.},
  volume  = {11},
  number  = {10},
  pages   = {646--650},
  year    = {2017},
  doi     = {10.1038/s41566-017-0007-1}
}

@article{Wang2021OneHour,
  author  = {Wang, Pengfei and others},
  title   = {{Single ion qubit with estimated coherence time exceeding one hour}},
  journal = {Nat. Commun.},
  volume  = {12},
  pages   = {233},
  year    = {2021},
  doi     = {10.1038/s41467-020-20330-w}
}

@article{Bradley2019,
  title = {{A Ten-Qubit Solid-State Spin Register with Quantum Memory up to One Minute}},
  author = {Bradley, C. E. and others},
  journal = {Phys. Rev. X},
  volume = {9},
  issue = {3},
  pages = {031045},
  numpages = {12},
  year = {2019},
  month = {Sep},
  publisher = {American Physical Society},
  doi = {10.1103/PhysRevX.9.031045},
  url = {https://link.aps.org/doi/10.1103/PhysRevX.9.031045}
}

@article{Bruzewicz2019,
    author = {Bruzewicz, Colin D. and Chiaverini, John and McConnell, Robert and Sage, Jeremy M.},
    title = {{Trapped-ion quantum computing: Progress and challenges}},
    journal = {Applied Physics Reviews},
    volume = {6},
    number = {2},
    pages = {021314},
    year = {2019},
    month = {05},
    issn = {1931-9401},
    doi = {10.1063/1.5088164},
    url = {https://doi.org/10.1063/1.5088164},
    eprint = {https://pubs.aip.org/aip/apr/article-pdf/doi/10.1063/1.5088164/19742554/021314_1_online.pdf},
}

@article{Koudia2025SpatialModeDiversity,
  title   = {{Spatial-mode diversity and multiplexing for continuous variables quantum communications}},
  author  = {Koudia, Seid and Oleynik, Leonardo and ur Rehman, Junaid and Chatzinotas, Symeon},
  journal = {Commun. Phys.},
  volume  = {8},
  number  = {351},
  year    = {2025},
  doi     = {10.1038/s42005-025-02258-z},
  publisher = {Nature Portfolio}
}

@article{ZhaoFSOCapacity2015,
  author  = {Zhao, Ningning and Li, Xiang and Li, Guifang and Kahn, Joseph M.},
  title   = {{Capacity limits of spatially multiplexed free-space communication channels}},
  journal = {Nat. Photon.},
  volume  = {9},
  number  = {12},
  pages   = {822--826},
  year    = {2015},
  doi     = {10.1038/nphoton.2015.214}
}

@article{PirandolaFSO2021,
  author  = {Pirandola, Stefano},
  title   = {{Limits and security of free-space quantum communications}},
  journal = {Phys. Rev. Res.},
  volume  = {3},
  number  = {1},
  pages   = {013279},
  year    = {2021},
  doi     = {10.1103/PhysRevResearch.3.013279}
}

@article{Bel2025SimulatorsReview,
  author    = {Oceane Bel and Mariam Kiran},
  title     = {{Simulators for quantum network modeling: A comprehensive review}},
  journal   = {Comput. Netw.},
  volume    = {263},
  pages     = {111204},
  year      = {2025},
  doi       = {10.1016/j.comnet.2025.111204},
  publisher = {Elsevier}
}

@article{El-Wakeel:16,
author = {Amr S. El-Wakeel and Nazmi A. Mohammed and Moustafa H. Aly},
journal = {Appl. Opt.},
number = {26},
pages = {7276--7286},
publisher = {Optica Publishing Group},
title = {{Free space optical communications system performance under atmospheric scattering and turbulence for 850 and 1550\&\#x2009;\&\#x2009;nm operation}},
volume = {55},
month = {Sep},
year = {2016},
url = {https://opg.optica.org/ao/abstract.cfm?URI=ao-55-26-7276},
doi = {10.1364/AO.55.007276},
}

@article{Singh2023,
    author = {Singh, Harjeevan and Mittal, Nitin and Ogudo, Kingsley A.},
    title = {{Optimizing the receiver aperture parameters of free space optical (FSO) link for performance enhancement}},
    journal = {AIP Conf. Proc.},
    volume = {2495},
    number = {1},
    pages = {020030},
    year = {2023},
    month = {10},
    issn = {0094-243X},
    doi = {10.1063/5.0123191},
    url = {https://doi.org/10.1063/5.0123191},
    eprint = {https://pubs.aip.org/aip/acp/article-pdf/doi/10.1063/5.0123191/18178540/020030_1_5.0123191.pdf},
}

@inproceedings{kolka2009simulation,
  title={{Simulation of atmospheric optical channel with ISI}},
  author={Kolka, ZDENEK and Biolek, Dalibor and Biolkova, VIERA},
  booktitle={CSECS},
  volume={9},
  pages={198--201},
  year={2009}
}

@article{Lushnikov2018Diffraction,
  author    = {Lushnikov, P. M. and Vladimirova, N.},
  title     = {{Toward Defeating Diffraction and Randomness for Laser Beam Propagation in Turbulent Atmosphere}},
  journal   = {JETP Lett.},
  volume    = {108},
  number    = {9},
  pages     = {571--576},
  year      = {2018},
  doi       = {10.1134/S0021364018210026},
}

@ARTICLE{Kaushal2017,
  author={Kaushal, Hemani and Kaddoum, Georges},
  journal={IEEE Commun. Surv. Tutor.}, 
  title={{Optical Communication in Space: Challenges and Mitigation Techniques}}, 
  year={2017},
  volume={19},
  number={1},
  pages={57-96},
  doi={10.1109/COMST.2016.2603518}}

@article{Morana2020UltraLowLossFiber,
  author    = {Morana, A. and others},
  title     = {{Extreme Radiation Sensitivity of Ultra-Low Loss Pure-Silica-Core Optical Fibers at Low Dose Levels and Infrared Wavelengths}},
  journal   = {Sensors},
  volume    = {20},
  number    = {24},
  pages     = {7254},
  year      = {2020},
  doi       = {10.3390/s20247254}
}

@article{Bouwmeester1997,
  author    = {D. Bouwmeester and J. W. Pan and K. Mattle and M. Eibl and H. Weinfurter and A. Zeilinger},
  title     = {{Experimental quantum teleportation}},
  journal   = {Nature},
  year      = {1997},
  volume    = {390},
  pages     = {575--579},
  doi       = {10.1038/37539},
  url       = {https://doi.org/10.1038/37539}
}

@article{bruss,
    author = {Bruß, Dagmar},
    title = {Characterizing entanglement},
    journal = {J. Math. Phys.},
    volume = {43},
    number = {9},
    pages = {4237-4251},
    year = {2002},
    month = {09},
    issn = {0022-2488},
    doi = {10.1063/1.1494474},
    url = {https://doi.org/10.1063/1.1494474},
    eprint = {https://pubs.aip.org/aip/jmp/article-pdf/43/9/4237/19182946/4237_1_online.pdf},
}

\end{document}